\documentclass[a4paper,UKenglish,modulo]{lipics-v2019}
\usepackage{microtype}
\usepackage{multirow}
\usepackage{subcaption}
\usepackage{amsfonts,amsmath,amsthm,nicefrac}
\usepackage{graphicx}
\usepackage{color}

\title{On the hardness of computing an average curve}

\author{Kevin Buchin}{Department of Mathematics and Computing Science, TU Eindhoven, The Netherlands}{k.a.buchin@tue.nl}{}{}
\author{Anne Driemel}{University of Bonn, Hausdorff Center for Mathematics, Bonn, Germany}{driemel@cs.uni-bonn.de}{}{}
\author{Martijn Struijs}{Department of Mathematics and Computing Science, TU Eindhoven, The Netherlands}{m.a.c.struijs@tue.nl}{}{}

\authorrunning{K. Buchin, A. Driemel and M. Struijs} 

\relatedversion{} 
\funding{}
\acknowledgements{}

\Copyright{Kevin Buchin, Anne Driemel, and Martijn Struijs}



\nolinenumbers 
\hideLIPIcs  

\newcommand{\rr}{\mathbb{R}}
\providecommand{\eps}{\varepsilon}

\DeclareMathOperator{\dtw}{DTW}
\newcommand*\frech{\mathop{}\!{\mathrm{d}_F}}
\newcommand*\dfrech{\mathop{}\!{\mathrm{d}_{dF}}}

\newcommand{\acount}{\#_A}
\newcommand{\bcount}{\#_B}

\newtheorem{problem}{Problem}

\begin{document}

\maketitle

\begin{abstract}
We study the complexity of clustering curves under $k$-median and $k$-center objectives in the metric space of the Fr\'echet distance and related distance measures. 
Building upon recent hardness results for the minimum-enclosing-ball problem under the Fr\'echet distance, we show that also the $1$-median problem is NP-hard.
Furthermore, we show that the $1$-median problem is W[1]-hard with the number of curves as parameter. We show this under the discrete and continuous Fr\'echet and Dynamic Time Warping (DTW) distance. This yields an independent proof of an earlier result by Bulteau et al.\ from 2018 for a variant of DTW that uses squared distances, where the new proof is both simpler and more general.
On the positive side, we give approximation algorithms for problem variants where the center curve may have complexity at most $\ell$ under the discrete Fr\'echet distance. In particular, for fixed $k,\ell$ and $\eps$, we give $(1+\eps)$-approximation algorithms for the $(k,\ell)$-median and $(k,\ell)$-center objectives and a polynomial-time exact algorithm for the $(k,\ell)$-center objective.
\end{abstract}

\section{Introduction}
Clustering is an important tool in data analysis, used to split data into groups of similar objects. Their dissimilarity is often based on distance between points in Euclidean space. However, the dissimilarity of polygonal curves is more accurately measured by specialised measures: Dynamic Time Warping (DTW)~\cite{PGAveragingSteiner}, continuous and discrete Fr\'echet distance~\cite{AltG95,eiter1994computing}.

We focus on \emph{centroid-based clustering}, where each cluster has a center curve and the quality of the clustering is based on the similarity between the center and the elements inside the cluster.  
In particular, given a distance measure $\delta$, we consider the following problems:
\begin{problem}[\textbf{$k$-median for curves with distance $\delta$}]\label{prob:k-median}
	Given a set $\mathcal{G}=\{g_1,\ldots,g_m\}$ of polygonal curves, find a set $\mathcal{C}=\{c_1,\ldots c_k\}$ of polygonal curves that minimizes $\sum\limits_{g\in\mathcal{G}} \min_{i=1}^k \delta(c_i,g)$.
\end{problem}
\begin{problem}[\textbf{$k$-center for curves with distance $\delta$}]\label{prob:k-center}
	Given a set $\mathcal{G}=\{g_1,\ldots,g_m\}$ of polygonal curves, find a set $\mathcal{C}=\{c_1,\ldots c_k\}$ of polygonal curves that minimizes $\max\limits_{g\in\mathcal{G}} \min_{i=1}^k \delta(c_i,g)$.	
\end{problem}

For points in Euclidean space, the most widely-used centroid-based clustering problem is $k$-means, in which the distance measure $\delta$ is the squared Euclidean distance.
But also for general metric spaces the $k$-median problem is well studied, often in the context of the closely related facility location problem~\cite{JainMS02,JainV01,LiS16}. In general metric spaces usually, the \emph{discrete $k$-median problem} is studied, where the centers must be selected from a finite set $F$, and are called facilities. 

For clustering curves, limiting the possible centers to a finite set of `facilities' is unnecessarily restrictive.
In this paper, we are therefore interested in the \emph{unconstrained $k$-median problem}, where a center can be any element of the metric space (as in the case of $k$-means). Often, we will simply write $k$-median problem to denote the unconstrained version. 
In this paper, we are in particular interested in the complexity of the 1-median problem, which we refer to  as \emph{average curve problem}.  

\paragraph*{Hardness of the average curve problem}
While clustering on points for general $k$ in the plane or higher dimension is often NP-hard~\cite{MegiddoS84}, many point clustering problems can be solved efficiently when $k=1$ in low dimension.
For instance, the $1$-center problem in the plane can be solved in linear time~\cite{Megiddo83a}, and there are practical algorithms for higher dimensional Euclidean space~\cite{FischerGK03}.
In contrast, the $1$-center problem (i.e., the minimum enclosing ball problem) for curves under the discrete and continuous Fr\'echet distance is already NP-hard in 1D~\cite{KAJFrechetCenter}.

In this paper, we show that also the average curve problem, i.e. the 1-median problem, is NP-hard. We show this for the discrete and for the continuous Fr\'echet distance, and for the dynamic time-warping (DTW) distance. Variants of the DTW distance differ in the norm used for comparing pairs of points, and how that norm is used, see Section~\ref{subsec:prelim} for details. Our results apply to a large class of variants of DTW. For the frequently used variant of DTW using the squared Euclidean distance, Bulteau~et~al.~\cite{bulteau2018hardness} recently showed that the average curve problem is NP-hard and even W[1]-hard when parametrized in the number of input curves $m$ and there exists no $f(m)\cdot n^{o(m)}$-time algorithm unless the Exponential Time Hypothesis (ETH) fails\footnote{See e.g.~\cite{Cygan2015Parametrized} for background on parametrized complexity}. Because of its importance in time series clustering, there are many heuristics for the average curve problem under DTW~\cite{gupta1996nonlinear,HautamakiNF08,PGAveragingSteiner}. Brill et al.\ showed that dynamic programming yields an exponential-time exact algorithm~\cite{Brill2019} and additionally show the problem can be solved in polynomial time when both the input curves and center curve use only vertices from $\{0,1\}$. 

\paragraph*{Approximation algorithms}

Since both the $k$-center and the $k$-median problem for curves are already NP-hard for $k=1$ in 1D, we further study efficient approximation algorithms for these problems.

For approximation in metric spaces, the discrete and unconstrained $k$-median (likewise for $k$-center) are closely related: any set of curves that realises an $\alpha$-approximation for the discrete $k$-median problem realises an $2\alpha$-approximation for the unconstrained $k$-median problem. There is an elegant $O(kn)$ time $2$-approximation algorithm for the $k$-center problem in metric spaces~\cite{Gonzalez85}. This approximation factor is tight for clustering curves under the discrete~\cite{KAJFrechetCenter} and continuous~\cite{Struijs2018Curve} Fr\'echet distance. Finding approximate solutions for $k$-median is more challenging: the best known polynomial-time approximation algorithm for discrete $k$-median in general metric space achieves a factor of $3+\eps$ for any $\eps>0$~\cite{AryaGKMMP04} and it is NP-hard to achieve an approximation factor of $1+2/e$~\cite{JainMS02}. 

Unconstrained clustering of curves may result in centers of high complexity. To avoid overfitting and to obtain a compact representation of the data, 
we look at a variant of the clustering problems with center curves of at most a fixed complexity, denoted by $\ell$. More formally, the \emph{$(k,\ell)$-center problem} is to find a set of curves $\mathcal{C}=\{c_1,\ldots c_k\}$, each of complexity at most $\ell$, that minimizes $\max_{g\in\mathcal{G}} \min_{i=1}^k \delta(c_i,g)$. The \emph{$(k,\ell)$-median problem} is defined analogously. 
Although the general case for this variant is still NP-hard, we can find efficient algorithms when $k$ and $\ell$ are fixed. The $(k,\ell)$-center and $(k,\ell)$-median problems were introduced by Driemel et~al.~\cite{driemel2016clustering}, who obtained an $\widetilde{O}(mn)$-time $(1+\eps)$-approximation algorithm for the $(k,\ell)$-center and $(k,\ell)$-median problem under the Fr\'echet distance for curves in 1D, assuming $k,\ell,\eps$ are constant. 
In~\cite{KAJFrechetCenter}, Buchin et~al.\ gave polynomial-time constant-factor approximation algorithms for the $(k,\ell)$-center problem under the discrete and continuous Fr\'echet distance for curves in arbitrary dimension. These approximation algorithms have lead to efficient implementations of heuristics for the center version showing that the considered clustering formulations are useful in practice~\cite{BDLNklcluster19}, see Figure~\ref{fig:klcluster}. This encourages further study of the median variants of the problem.

\begin{figure}
	\centering
	\includegraphics[scale=0.1]{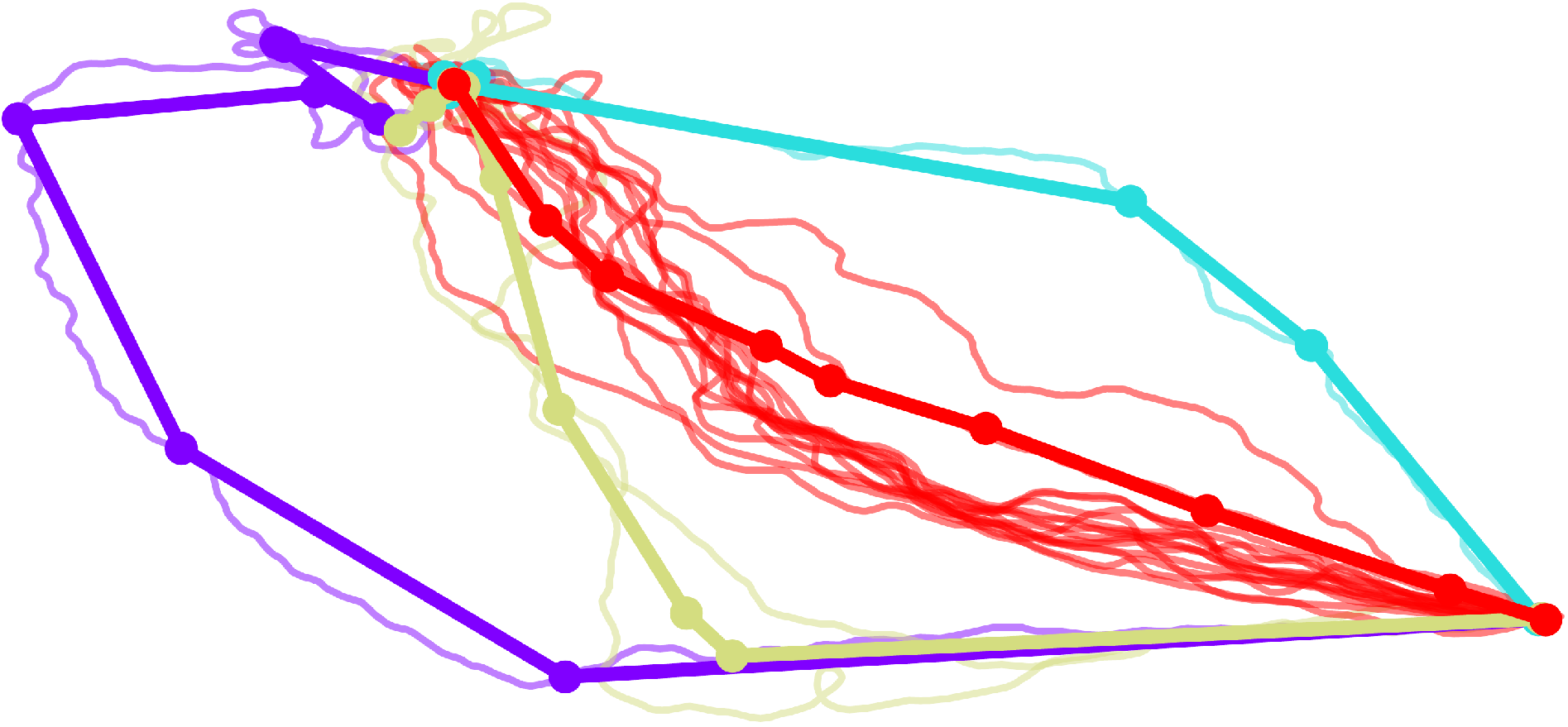}
	\caption{\label{fig:klcluster} $(k,\ell)$-center clustering of pigeon flight paths computed by the algorithm of~\cite{BDLNklcluster19}.}
\end{figure}

\subsection{Definitions of distance measures}\label{subsec:prelim}
Let $x$ be a polygonal curve, defined by a sequence of vertices $x_1,\ldots,x_n$ from $\rr^d$ where consecutive vertices are connected by straight line segments. We call the number of vertices of $x$ the \emph{complexity}, denoted by $|x|$.  Given a pair of polygonal curves $x,y$, a \emph{warping path} between them is a sequence $W =\langle w_1,\ldots, w_L\rangle$ of index pairs $w_l=(i_l,j_l)$ from $\{1,\ldots, |x|\}\times \{1,\ldots, |y|\}$ such that $w_1 = (1,1)$, $w_L = (|x|,|y|)$, and $(i_{l+1}-i_l,j_{l+1}-j_l) \in \{(0,1),(1,0),(1,1)\}$ for all $1\leq l <L$. We say two vertices $x_i$, $y_i$ are \emph{matched} if $(i,j)\in W$. 

Denote the set of all warping paths between curves $x$ and $y$ by $\mathcal{W}_{x,y}$. For any integers $p,q\geq 1$, we define the Dynamic Time Warping Distance between $x$ and $y$ as
\[\dtw_p^q(x,y):= \left(\min_{W\in \mathcal{W}_{x,y}} \sum_{(i,j)\in W} \|x_i - y_j\|^p\right)^{q/p},\] where $\|\cdot\|$ denotes the Euclidean norm. In text, we refer to $\dtw_p^q$ also as $(p,q)$-DTW. Similarly, define the discrete Fr\'echet distance between $x,y$ as 
\[ \dfrech(x,y):= \min_{W\in \mathcal{W}_{x,y}} \max_{(i,j)\in W} \|x_i - y_j\|.\]

The \emph{continuous} Fr\'echet distance is defined with a reparametrization $f:[0,1]\rightarrow [0,1]$, which is a continuous injective function with $f(0)=0$ and $f(1)=1$. We say two points on $x$ and $y$ are \emph{matched} if $f(i) = j$. Denote the set of all reparametrizations by $\mathcal{F}$, then the continuous Fr\'echet distance is given by 
\[ \frech(x,y):= \inf_{f\in \mathcal{F}} \max_{\alpha\in [0,1]} \|x(f(\alpha)) - y(\alpha)\|.\]

\begin{table}[t]
	\centering
	\caption{Overview of results. In these tables, $n$ denotes the length of the input curves, $m$ denotes the number of input curves and $d$ denotes the ambient dimension of the curves. }\label{tab:overview}
	\subcaption{Results on exact computation.}
	\begin{tabular}{|c|c|c|c|}
		\hline
		Problem & Result & Restrictions & Reference  \\
		\hline\hline
		\multirow{6}{0.22\textwidth}{1-median, DTW$^q_p$}
		&$O(n^{2m+1} 2^{m} m)$ & $d=1$ & \multirow{2}{*}{Brill et al.~\cite{Brill2019}}\\
		&$O(mn^3)$ & Binary & \\
		\cline{2-4}
		&NP-hard & \multirow{2}{*}{$p=q=2$} & \multirow{2}{*}{Bulteau et al.~\cite{bulteau2018hardness}}\\
		&W[1]-hard in $m$ & & \\
		\cline{2-4}
		&NP-hard &  \multirow{2}{*}{$p,q \in \mathbb{N}$} & \multirow{2}{*}{Theorem~\ref{thm:DTW-hard}}\\
		&W[1]-hard in $m$ & & \\
		\hline\hline
		\multirow{2}{0.22\textwidth}{1-median, Fr\'echet} 
		& NP-hard  & & \multirow{2}{*}{Theorem~\ref{thm:Frechet-hard}}  \\
		& W[1]-hard in $m$ & &  \\
		\hline\hline
		\multirow{2}{0.22\textwidth}{$1$-center, discrete~Fr\'echet}
		& \multirow{2}{*}{NP-hard} & & \multirow{2}{*}{Buchin et al.~\cite{KAJFrechetCenter}}\\
		& & & \\
		\hline
		\multirow{2}{0.22\textwidth}{$(k,\ell)$-center, discrete~Fr\'echet}
		& \multirow{2}{*}{$O((mn)^{2k\ell}k\ell m \log (mn))$ } 
		& \multirow{2}{*}{$d \leq 2$} 
		& \multirow{2}{*}{Theorem~\ref{thm:Frechet-exact}}\\
		& & & \\
		\hline
	\end{tabular}\bigskip
	\subcaption{Approximation algorithms. (In stating the running times  we assume $k$, $\ell$, and $\eps$ are constants independent of $n$ and $m$.) }
	\begin{tabular}{|c|c|c|c|c|}
		\hline
		Problem & Result & Approximation factor & Restrictions & Reference  \\ 
		\hline\hline 
		\multirow{2}{0.22\textwidth}{$(k,\ell)$-median, continuous Fr\'echet }
		&\multirow{2}{*}{$\widetilde{O}(nm)$} 
		&\multirow{2}{*}{$(1+\eps)$} 
		&\multirow{2}{*}{$d=1$}
		&\multirow{2}{*}{Driemel et al.~\cite{driemel2016clustering}} \\
		& & & & \\
		\hline
		\multirow{4}{0.22\textwidth}{$(k,\ell)$-median, discrete~Fr\'echet}
		& $\widetilde{O}(nm)$ & $65$ & & Driemel et al.~\cite{driemel2016clustering} \\
		& $\widetilde{O}(m^2(m+n))$ & $12$ & & Theorem~\ref{thm:12-approx-k,l-median} \\
		& $\widetilde{O}(nm)$ & $(1+\eps)$ & $k=1$ & Theorem~\ref{thm:Frechet-median-approx} \\
		& $\widetilde{O}(nm^{dk\ell+1})$ & $(1+\eps)$ & $k>1$ & Theorem~\ref{thm:Frechet-median-approx} \\
		\hline\hline 
		\multirow{2}{0.22\textwidth}{$(k,\ell)$-center, discrete~Fr\'echet} 
		& $\widetilde{O}(nm)$ & 3 & & Buchin et al.~\cite{KAJFrechetCenter}\\
		&{$\widetilde{O}(nm)$} &$(1+\eps)$ 
		& & Theorem~\ref{thm:Frechet-center-approx}\\
		\hline
	\end{tabular}
\end{table}

\subsection{Results}
We show that the average curve problem for discrete and continuous Fr\'echet distance in 1D is NP-complete, W[1]-hard when parametrized in the number of curves $m$, and admits no $f(m)\cdot n^{o(m)}$-time algorithm unless ETH fails.
In addition, we prove the same hardness results of the average curve problem for the $(p,q)$-DTW distance for any $p,q\in \mathbb{N}$. 

This is an independent proof that is simpler and more general than the result by Bulteau~et~al.~\cite{bulteau2018hardness}. Their hardness result holds for the case of the $(2,2)$-DTW distance, which is widely-used. Other common variants, covered by our proof, are $(1,1)$-DTW, i.e., (non-squared) Euclidean distance and Manhattan distance in 1D~\cite{JSSv031i07}, $(2,1)$-DTW, and more generally $(p,1)$-DTW~\cite{dtwclust-17,RJ-2019-023}. Note that, while we define $(p,1)$-DTW in terms of the $p$th power of the Euclidean norm, our hardness results also apply to the $p$th power of the $L_p$-norm, since these norms are equal in 1D.

Brill~et~al.~\cite{Brill2019} asked whether their result can be extended to obtain a polynomial time algorithm when all curves are restricted to sets of $3$ vertices. In our NP-hardness construction, both the input curves and the center curve use only vertices from $\{-1,0,1\}$, so we answer this question in the negative, unless $\textsc{P}=\textsc{NP}$. (The hardness construction by Bulteau~et~al.\ uses a center curve that is not restricted to a bounded set of vertices, and therefore does not resolve this question)

Since our and other hardness results exclude efficient algorithms for the $(k,\ell)$-center or -median clustering without further assumptions, we investigate other approaches with provable guarantees. In particular, we give a $(1+\eps)$-approximation algorithm that runs in $\widetilde{O}(mn)$ time and a polynomial-time exact algorithm to solve the $(k,\ell)$-center problem for the discrete Fr\'echet distance, when $k,\ell$, and $\eps$ are fixed. For the $(k,\ell)$-median problem under the discrete Fr\'echet distance, we give a polynomial time $12$-approximation algorithm, and an $(1+\eps)$-approximation algorithm that runs in polynomial time when $k,\ell$, and $\eps$ are fixed.
Table~\ref{tab:overview} gives an overview of our results.
\section{Hardness of the average curve problem for discrete and continuous Fr\'echet}

In this section, we will show that the $1$-median problem (or average curve problem) is NP-hard for the discrete and continuous Fr\'echet distance. The average curve problem for the discrete Fr\'echet distance is as follows: given a set of curves $\mathcal{G}$  and an integer $r$, determine whether there exists a center curve $c$ such that $\sum_{g\in\mathcal{G}}\dfrech(c,g)\leq r$. We will show that this problem is NP-hard. To find a reasonable algorithm, we can look at a parametrized version of the problem. A natural parameter is the number of input curves, which we will denote by $m$. However, we will show that this parametrized problem is W[1]-hard, which rules out any $f(m)\cdot n^{O(1)}$-time algorithm, unless $\textsc{FPT}=\textsc{W}[1]$. To achieve these reductions, we create a reduction from a variant of the shortest common supersequence (SCS) problem. 

\subsection{The FCCS problem}

To show the hardness of the average curve problem for the Fr\'echet and DTW distance, we reduce from a variant of the \emph{Shortest Common Supersequence} (SCS) problem, which we will call the \emph{Fixed Character Common Supersequence} (FCCS) problem. If $s$ is a string and $x$ is a character, $\#_x(s)$ denotes the number of occurrences of $x$ in $s$. 

\begin{problem}[\textbf{Shortest Common Supersequence (SCS)}] Given a set $S$ of $m$ strings with length at most $n$ over the alphabet $\Sigma$ and an integer $t$, does there exists a string $s^*$ of length $t$ that is a supersequence of each string $s\in S$?
\end{problem}
 
\begin{problem}[\textbf{Fixed Character Common Supersequence (FCCS)}] Given a set $S$ of $m$ strings with length at most $n$ over the alphabet $\Sigma=\{A,B\}$ and $i,j\in\mathbb{N}$, does there exists a string $s^*$ with $\acount(s^*)=i$ and $\bcount(s^*)=j$ that is a supersequence of each string $s\in S$?
\end{problem}

The SCS problem with a binary alphabet is known to be NP-hard~\cite{RAIHA1981187} and $W[1]$-hard~\cite{pietrzak2003parameterized}. The same holds for our variant:  

\begin{lemma}\label{lem:SCS-prime}
	The FCCS problem is NP-hard. The FCCS problem with $m$ as parameter is W[1]-hard. There exists no $f(m)\cdot n^{o(m)}$ time algorithm for FCCS unless ETH fails. 
\end{lemma}
\begin{proof}
We reduce from SCS with the binary alphabet $\{A,B\}$ to FCCS. Given an instance $(S,t)$ of SCS, construct $S' = \{s+AB^{2t}A+c(s)\mid s\in S\}$, where $c(s)$ denotes the string constructed by replacing all A characters in $s$ by B and vice versa, and $+$ denotes string concatenation.  
We reduce to the instance $(S',t+2,3t)$ of FCCS and claim that $(S,t)$ is a true instance of SCS if and only if $(S',t+2,3t)$ is a true instance of FCCS. 

If $(S,t)$ is a true instance of SCS, then there exists a string $q$ of length $t$ that is a supersequence of each string in $S$. Therefore, the string $q'= q+ AB^{2t}A+ c(q)$ is a supersequence of all strings in $S'$. Since $\acount(q') = 2 + \acount(q+c(q)) = 2+t$ and $\bcount(q') = 2t + \bcount(q+c(q)) = 3t$, $(S', t+2,3t)$ is a true instance of FCCS.

If $(S',t+2,3t)$ is a true instance of FCCS, there is string $q'$ with $\acount(q')=t+2$ and $\bcount(q')=3t$ that is a supersequence of each string $s'\in S'$. Consider a pair of strings $s_1' = s_1 + AB^{2t}A + c(s_1)$ and $s_2' = s_2 + AB^{2t}A + c(s_2)$ from $S'$. If there is no matching such that the first character of the $AB^{2t}A$ substring in $s_1'$ is matched to the same character of $q'$ as the first character of that substring in $s_1'$, then $q'$ is a supersequence of $AB^{2t}AB^{2t}A$ and so $\bcount(q')> 3t$, a contradiction.
By symmetry, the same holds for the last character of the substring $AB^{2t}A$ and therefore $q=q_1+q_2+q_3$, where $q_1$ is a supersequence of $S$, $q_2$ is a supersequence of $AB^{2t}A$ and $q_3$ is a supersequence of $\{c(s)\mid s\in S\}$. Note that $c(q_3)$ is a supersequence of $S$. Also, $\acount(q_1)+\acount(c(q_3))= \acount(q)- \acount(q_2) \leq t$ and $\bcount(q_1)+\bcount(c(q_3)) = \bcount(q)-\bcount(q_2)\leq t$. So, $|q_1|+|c(q_3)|\leq 2t$, which means that $|q_1|\leq t$ or $|c(q_3)|\leq t$ and thus $(S,t)$ is a true instance of SCS.

Note that this reduction is both a polynomial-time reduction and a parametrized reduction in the parameter $m$. Since the SCS problem over the binary alphabet $\{A,B\}$ is NP-hard~\cite{RAIHA1981187} and W[1]-hard when parametrized with the number of strings $m$~\cite{pietrzak2003parameterized}, the first two parts of the claim follow. The final part of the claim follows from the fact that this reduction 

Together with the reduction from~\cite{pietrzak2003parameterized}, we have a parametrized reduction from \textsc{Clique} with a linear bound on the parameter, so the final part of the claim follows~\cite[Obs. 14.22]{Cygan2015Parametrized}.
\end{proof}

\subsection{Complexity of the average curve problem under the discrete and continuous Fr\'echet distance}
We will show the hardness of finding the average curve under the discrete and continuous Fr\'echet distance via the following reduction from FCCS. Given an instance $(S,i,j)$ of FCCS, we construct a set of curves using the following vertices in $\rr$: $g_a=-1$, $g_b=1$, $g_A=-3$, and $g_B=3$. For each string $s\in S$, we map each character to a subcurve in $\rr$:
\[	A\rightarrow (g_ag_b)^{i+j}g_A(g_bg_a)^{i+j} \qquad B\rightarrow (g_bg_a)^{i+j}g_B(g_ag_b)^{i+j}.\]
The curve $\gamma(s)$ is constructed by concatenating the subcurves resulting from this mapping, $G=\{\gamma(s)\mid s\in S\}$ denotes the set of these curves. Additionally, we use the curves
\[A_i= g_b(g_Ag_b)^i\qquad B_j = g_a(g_Bg_a)^j.\]

We will call subcurves containing only $g_A$ or $g_B$ vertices \emph{letter gadgets} and subcurves containing only $g_a$ or $g_b$ vertices \emph{buffer gadgets}.  Let $R_{i,j}=\{A_i, B_j\}$. We reduce to the instance $(G\cup R_{i,j},r)$ of the average curve problem, where $r=|S|+2$. We use the same construction for the discrete and continuous case.

For an example of this construction, take $S=\{ABB,BBA,ABA\}$, $i=2$, $j=2$. Then $ABBA$ is a supersequence of $S$ with the correct number of characters. Note that the curve with vertices $0g_A0g_B0g_B0g_A0$ has a (discrete) Fr\'echet distance of at most $1$ to the curves in $G\cup R_{i,j}$, see Figure~\ref{fig:frechet-median-hardness}, so the sum of those distances is at most $|S|+2=r$.

\begin{figure}[p]
	\centering
	\includegraphics[scale=0.7,page=2]{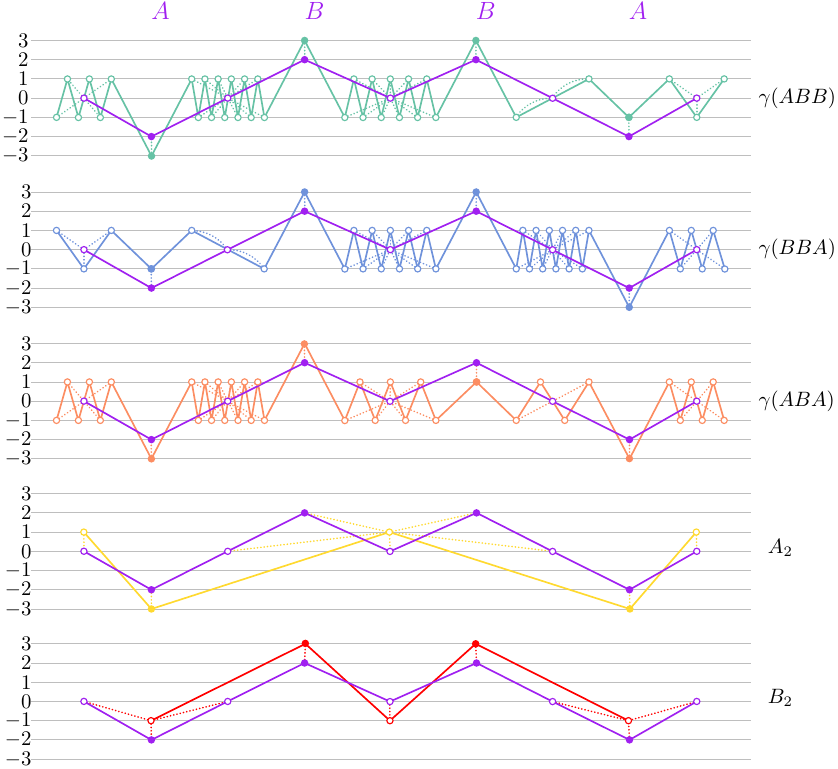}
	\caption{\label{fig:frechet-median-hardness} Five curves from $G\cup R_{i,j}$ in the reduction for the Fr\'echet average curve problem and a center curve constructed from $ABBA$ (purple) as in Lemma~\ref{lem:FCCS-to-curve-frechet}. Matchings are indicated by dotted lines. Note that each of these matchings achieves a (discrete) Fr\'echet distance of $1$.}
\end{figure}
\begin{figure}[p]
	\centering
	\includegraphics[scale=0.7]{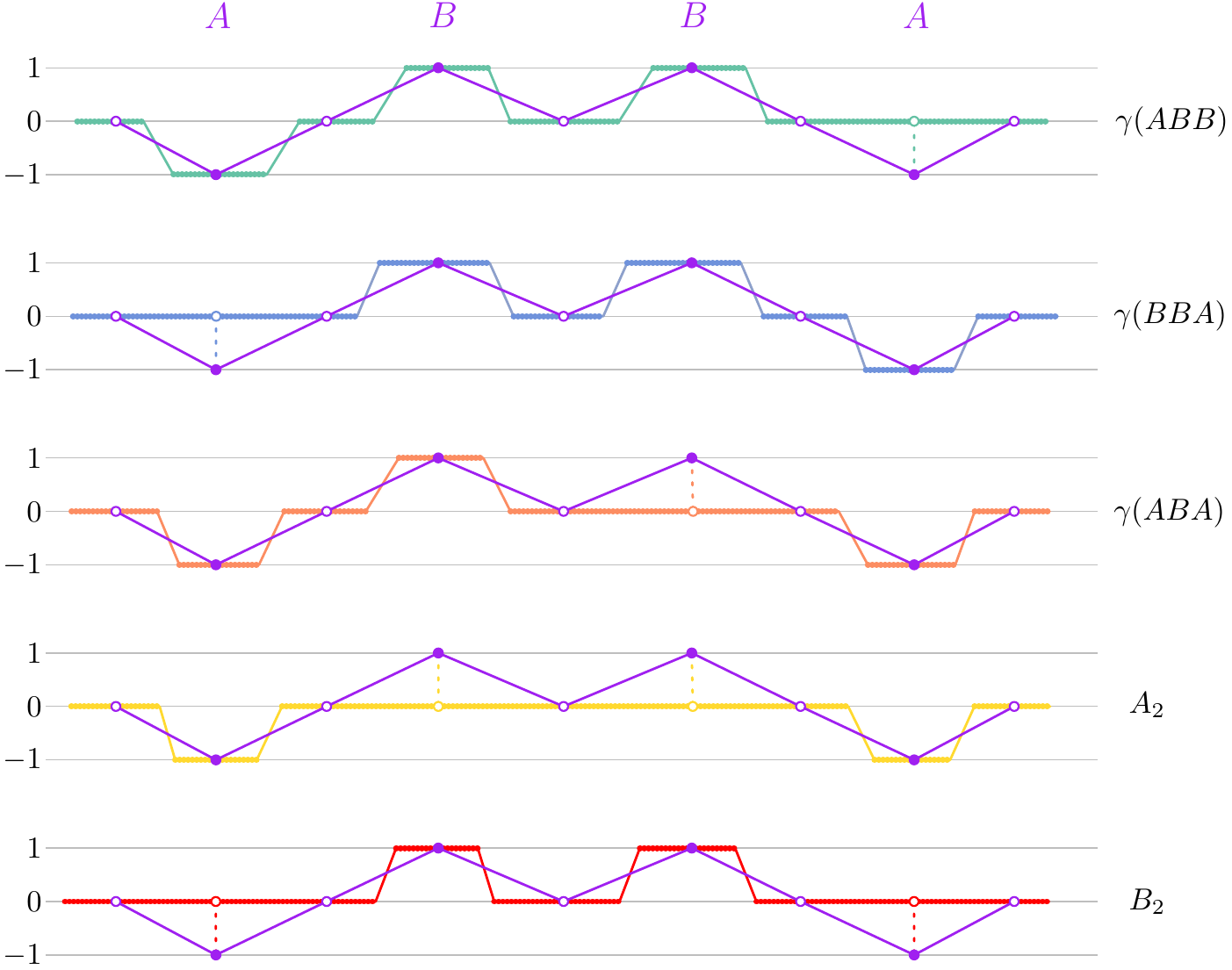}
	\caption{\label{fig:DTW-median-hardness} Five curves from $G\cup R_{i,j}$ in the reduction for the DTW average curve problem and a center curve constructed from the string $ABBA$ (purple) as in Lemma~\ref{lem:FCCS-to-DTW}. Fat horizontal lines indicate $\beta$ consecutive vertices. Vertices that match at distance $0$ touch, vertices that match at distance $1$ are indicated by dotted lines. The center has $1$ mismatch with the first $3$ curves and $2$ with the final two, so the total cost here is $3\cdot (1^p)^{q/p} + 2\alpha\cdot (2\cdot 1^p)^{q/p} = 3+2\alpha \cdot 2^q$  }
\end{figure}

\begin{lemma}\label{lem:FCCS-to-curve-frechet}
	If $(S,i,j)$ is a true instance of FCCS, then $(G\cup R_{i,j}, r)$ is a true instance of the average curve problem for discrete and continuous Fr\'echet.
\end{lemma}
\begin{proof}
	We will show the proof for the discrete Fr\'echet distance. Since the discrete Fr\'echet distance is an upper bound of the continuous version, this proves the continuous case as well.
	
	Since $(S,i,j)$ is a true instance of FCCS, there exists a common supersequence $s^*$ of $S$ with $\acount(s^*)=i$ and $\bcount(s^*)=j$. Construct the curve $c$ of complexity $2|s^*|+1$, given by
	\[c_l=\begin{cases} 0 & \text{if $l$ is odd}\\ -2 & \text{if $l$ is even and $s^*_{l/2}=A$}\\ 2 & \text{if $l$ is even and $s^*_{l/2}=B$} \end{cases}, \] for each $l\in \{1,\ldots,2|s^*|+1\}$. Let $s\in S$, then note that the sequence of letter gadgets in $\gamma(s)$ is a subsequence of the letter gadgets in $c$, because $s$ is a subsequence of $s^*$. So, all letter gadgets in $\gamma(s)$ can be matched with a letter gadget in $c$, the remaining letter gadgets in $c$ with a buffer gadget in $\gamma(s)$ and all remaining buffer gadgets with another buffer gadget, such that $\dfrech(c,\gamma(s))\leq 1$. For the matching with $A_i$, note that $c$ has exactly $i$ $g_A$ vertices, so these can be matched with the $i$ $g_A$ vertices in $A_i$. All other vertices in $c$ have distance $1$ to the remaining buffer gadgets in $A_i$, so $\dfrech(c,A_i)\leq 1$. Analogously, $\dfrech(c,B_j)\leq 1$. So, we get $\sum_{g\in G\cup R_{i,j}} \dfrech(c,g) = \sum_{s\in S} \dfrech(c,\gamma(s)) + \dfrech(c,A_i)+\dfrech(c,B_j) \leq |S|+2 = r$, and $(G\cup R_{i,j}, r)$ is a true instance of average curve for discrete Fr\'echet.
\end{proof}

\begin{lemma}
\label{lem:curve-frechet-to-FCCS}
If $(G\cup R_{i,j}, r)$ is a true instance of the average curve problem for discrete and continuous Fr\'echet, then $(S,i,j)$ is a true instance of FCCS.\end{lemma}
\begin{proof}	
	We will show the proof for the continuous Fr\'echet distance. Since the continuous Fr\'echet distance is a lower bound of the discrete version, this proves the discrete case as well. 
	
	We call the interval of points $p$ on a subcurve $g_bg_Ag_b$ with $p<-1$ an \emph{A-peak}, and the interval of points $p$ on a subcurve $g_ag_Bg_a$ with $p>1$ a \emph{B-peak}. A curve $\gamma(s)$ has exactly one peak for every letter in $s$.
	
	Since $(G\cup R_{i,j}, r)$ is a true instance of the average curve problem for continuous Fr\'echet, there exists a curve $c^*$ such that $\sum_{g\in G\cup R_{i,j}} \frech(c^*,g)\leq r = |S|+2$. We start by deriving bounds for the distance between $c^*$ and the individual curves in $G\cup R_{i,j}$. 
	\begin{claim*} $\frech(\gamma(s),\gamma(s'))\geq 2$ for all $s,s' \in S$ such that $s\neq s'$.
	\end{claim*}
	\begin{proof}
		If a letter vertex $p$ on $\gamma(s)$ is matched with a point $p'$ that does not lie on a peak of the same letter in $\gamma(s')$, then $|p-p'|\geq 2$ and so $\frech(\gamma(s),\gamma(s'))\geq 2$. By symmetry, the same holds if we exchange $s$ and $s'$.

		Otherwise, each letter vertex can be matched only with points on a peak of the same letter. Let $k$ be the first index such that $s[k]\neq s'[k]$. Then, the $k$-th letter vertex of $\gamma(s)$ cannot be matched to any point on the $k$-th peak of $\gamma(s')$ and must be matched to a point on another peak; the same holds with $s$ and $s'$ exchanged. It is not possible that on both curves the $k$-th letter vertex is matched with a peak of index larger than $k$, since the matching is monotone. So, one of the curves has its $k$-th letter vertex matched with a point on a peak of index smaller than $k$, we assume w.l.o.g. that this curve is $s$.
		
		By monotonicity, the first $k$ letter vertices of $s$ are matched to the first $k-1$ peaks of $s'$, so there are two letter vertices on $s$ that are both matched with a point on the same peak on $s'$. The interval between those two points on this peak on $s'$ must be matched with the interval between the letter vertices on $s$, so all points in the buffer gadget between the letter vertices on $s$ are matched to some point on the peak on $s'$. But then there is either a point on an A-peak matched to $g_b$ or a point on a B-peak matched to $g_a$, which in both cases has distance a least $2$, so $\frech(\gamma(s),\gamma(s'))\geq 2$.
	\end{proof}

    \begin{claim*}
    $\frech(c^*,A_i) + \frech(c^*,B_j) \leq 2$
    \end{claim*}
    \begin{proof}
	Using the previous claim and the triangle inequality, we have 
	\[d_F(c^*,\gamma(s_k)) + d_F(c^*,\gamma(s_{k+1}))\geq d_F(\gamma(s_{k}),\gamma(s_{k+1})\geq 2\] for all $k\in \{1,\ldots, m-1\}$ and $d_F(c,\gamma(s_m)) + d_F(c,\gamma(s_1))\geq 2$. The summation of these $m$ inequalities has each $s_k$ exactly twice on the lefthand side, so $\sum_{k=1}^m 2d_F(c^*,\gamma(s_k)) >= 2m$, hence $\sum_{k=1}^m d_F(c^*, \gamma(s_k)) \geq m = |S|$. 
	So, $\frech(c^*,A_i) + \frech(c^*,B_j) \leq r - \sum_{k=1}^m d_F(c^*, \gamma(s_k)) \leq 2$.
	\end{proof}

    \begin{claim*} $\frech(c^*,A_i)\geq 1$ and $\frech(c^*,B_j)\geq 1$.
    \end{claim*}
	\begin{proof}
		Suppose $\frech(c^*,A_i)<1$. Then, all points $p$ on $c^*$ are matched to some point in $[-3,1]$ with distance $<1$, which means $|p-g_B|>1$. We can assume that each string in $S$ contains at least one $B$ character (if there is a string $s$ with only $A$ characters, any supersequence with $i$ A-characters is a supersequence of $s$ when $|s|\leq i$ and none when $|s|>i$, so we can remove such trivial strings from the instance and check if the instance is trivially false). Therefore, $\frech(c^*,\gamma(s))>1$ for any $s\in S$. 
		
		Since $|g_a-g_b|=2$, we have $\frech(A_i,B_j)\geq 2$, so $\frech(c^*,A_i) + \frech(c^*,B_j)\geq \frech(A_i,B_j)\geq 2$. But then $r\geq \sum_{g\in G\cup R_{i,j}} \frech(c^*,g) > |S|+2 = r$, a contradiction, so $\frech(c^*,A_i)\geq 1$.  The proof of $\frech(c^*,B_j)\geq 1$ is analogous.
	\end{proof}
	
	\begin{claim*}
	$\frech(c^*,g)=1$ for all $g\in G\cup R_{i,j}$.
	\end{claim*}
	\begin{proof}
	The last two claims together imply $\frech(c^*,A_i)= \frech(c^*,B_j)= 1$. This means that for each point $p$ on $c^*$, $|p|\leq 2$ (otherwise, $p$ has distance $>1$ to all points on $A_i$ or all points on $B_j$), so $\frech(c^*,\gamma(s))\geq 1$ for all $s\in S$, since we can assume $s$ contains at least one $A$ and $B$ character. Therefore, $\frech(\gamma(s),c^*)\leq r - \frech(A_i,c^*)-\frech(B_j,c^*) - \sum_{s'\in S\setminus\{s\}} \frech(\gamma(s'),c^*) \leq |S|- (|S|-1)= 1$ for all $s\in S$. 
	\end{proof}
	
	Now we have shown that any center curve that achieves a cost of $|S|+2$ for the constructed $k$-median instance needs to have Fr\'echet distance equal to $1$ to all curves in this instance. It remains to show that such a center curve encodes a solution to the initial FCCS instance. Note that such a center curve is also a solution to the $1$-center problem for this set of curves. We can now apply the proof of Lemma 33 from~\cite{KAJFrechetCenter, KAJFrechetCenterArXiv}, where the same gadgets were used in the reduction to the $1$-center problem.	
\end{proof}

\begin{theorem}\label{thm:Frechet-hard}
	The average curve problem for discrete and continuous Fr\'echet distance is NP-hard. When parametrized in the number of input curves $m$, this problem is W[1]-hard. There exists no $f(m)\cdot n^{o(m)}$ time algorithm for this problem unless ETH fails.
\end{theorem}
\begin{proof} 
	By Lemmas~\ref{lem:FCCS-to-curve-frechet} and~\ref{lem:curve-frechet-to-FCCS}, we have a valid reduction from FCCS to the average curve problem.
	Since this reduction runs in polynomial time and FCCS is NP-hard (Lemma~\ref{lem:SCS-prime}), the average curve problem for discrete and continuous Fr\'echet is NP-hard. Note that the number of curves in the reduced average curve instance is $k+2$, where $k$ is the number of input sequences of the FCCS instance. So, together with the reduction from Lemma~\ref{lem:SCS-prime}, this reduction is also a parametrized reduction from \textsc{Clique} with a linear bound on the parameter to the average curve problem for discrete and continuous Fr\'echet with the number of curves as a parameter, which implies the remainder of the theorem.
\end{proof}

\section[Hardness of the average curve problem for (p,q)-DTW]{Hardness of the average curve problem for $(p,q)$-DTW}

We will show that the average curve problem under the $(p,q)$-DTW distance is NP-hard for all $p,q\in \mathbb{N}$. This generalises the result of~\cite{bulteau2018hardness}, who use different methods to achieve the same hardness results for the $(2,2)$-DTW average curve problem only. We again reduce from FCCS instance $(S,i,j)$. Given a string $s\in S$ over the binary alphabet $\{A,B\}$, we map each character to a subcurve in $\rr$: 
\[A\rightarrow  g_0^\beta g_a^\beta g_0^\beta \qquad B\rightarrow g_0^\beta g_b^\beta g_0^\beta,\] where $g_0=0, g_a=-1, g_b=1$ as before and $\beta$ is a large constant that will be determined later. The curve $\gamma(s)$ is constructed by concatenating these subcurves and $G=\{\gamma(s)\mid s\in S\}$. We additionally use the curves \[A_i = g_0^\beta (g_a^\beta g_0^\beta)^i \qquad B_j= g_0^\beta (g_b^\beta g_0^\beta)^j.\] Call any subcurve consisting of $g_a$ or $g_b$ vertices a letter gadget and any subcurve consisting of $g_0$ a buffer gadget. Let $R_{i,j}$ contain curves $A_i$ and $B_j$, both with multiplicity $\alpha$. We reduce to the instance $(G\cup R_{i,j},r)$ of $(p,q)$-DTW average curve, where $r=\sum_{s\in S}(i+j-|s|)^{q/p} + \alpha (i^{q/p}+j^{q/p})$, $\beta=\lceil r/\eps^{q}\rceil +1$, $\alpha=|S|$ and $\eps = 1-(1-\min_{x\in \{i,j\}} \frac{(x+1)^{q/p}-x^{q/p}}{4(i+j)^{q/p}})^{1/q}$. See Figure~\ref{fig:DTW-median-hardness} for an example of this construction with $S=\{ABB,BBA,ABA\}$ and $i=j=2$.

The following definitions are used to prove Lemma~\ref{lem:pq-dtw-to-FCCS}. Take a vertex $p$ on some center curve $c^*$. If $|p-g_a|<\eps$, we call $p$ an \emph{A-signal vertex}. If $|p-g_b|<\eps$ we call $p$ an \emph{B-signal vertex}. If $p$ is not a signal vertex, then we call $p$ a \emph{buffer vertex}. Note that $\eps$ is chosen small enough such that no vertex is both an A- and B-signal vertex. We will show that the sequence of signal vertices in the curve satisfying $(G\cup R_{i,j},r)$ is a supersequence satisfying $(S,i,j)$.

\begin{lemma}\label{lem:FCCS-to-DTW}
	If $(S,i,j)$ is a true instance of FCCS, then $(G\cup R_{i,j},r)$ is a true instance of $(p,q)$-DTW average curve.
\end{lemma}
\begin{proof} If $(S,i,j)$ is a true instance of FCCS, then there exists a string $s^*$ that is a supersequence of $S$, with $\acount(s^*)=i$ and $\bcount(s^*)=j$. Construct the curve $c$ of length $2(i+j)+1$: \[c_l=\begin{cases} 0 & \text{if $l$ is odd}\\ g_a & \text{if $l$ is even and $s^*_{l/2}=A$}\\ g_b & \text{if $l$ is even and $s^*_{l/2}=B$} \end{cases}, \] for each $l\in \{1,\ldots,2(i+j)+1\}$. Analogously to Lemma~\ref{lem:FCCS-to-curve-frechet}, we can match the letter gadgets from $\gamma(s)$ to $g_A$ or $g_B$ in $c$ as $s^*$ is a supersequence of $s$, the letter gadgets of $A_i,B_j$ to $g_A,g_B$ in $c$ as the number of curves match, and $g_0$ vertices to buffer gadgets. This gives a matching such that $\sum_{g\in G\cup R_{i,j}} \dtw_p^q(c,g) \leq r$.
\end{proof}

\begin{lemma}\label{lem:pq-dtw-to-FCCS}
	If $(G\cup R_{i,j},r)$ is a true instance of $(p,q)$-DTW average curve, then $(S,i,j)$ is a true instance of FCCS. 
\end{lemma}
\begin{proof}
	If $(G\cup R_{i,j},r)$ is a true instance of $(p,q)$-DTW average curve, then there exists a curve $c^*$ such that $\sum_{g\in G\cup R_{i,j}} \dtw_p^q(c^*,g)\leq r$. Take a curve $g\in G\cup R_{i,j}$. First note that there is at least one signal vertex in $c^*$ matched to each letter gadget in $g$: otherwise, matching all $\beta$ vertices in the gadget costs at least $\eps^q\cdot \beta = \eps^q\cdot (r/\eps^q+1)>r$, which contradicts the choice of $c^*$. Similarly, each signal vertex is matched to at most one letter gadget in $g$, since otherwise it would have to match a $g_0^\beta$ subcurve in between the letter gadgets, which would have a cost of at least $(1-\eps)^q\cdot \beta > \eps^q\cdot \beta  >r$. This means that the sequence of letter gadgets in $\gamma(s)$ is a subsequence of the sequence of signal vertices in $c^*$. So, if we construct $s'$ from the sequence of signal vertices in $c^*$ by mapping A-signal vertices to $A$ characters and B-signal vertices to $B$ characters, we have that $s'$ is a supersequence of $S$.  What remains to be proven is that $\acount(s')=i$ and $\bcount(s')=j$, i.e. there are exactly $i$ A-signal vertices and $j$ B-signal vertices.
	
	First, note that the sequence of A letter gadgets in $A_i$ is a subsequence of the sequence of signal vertices in $c^*$ (using the same argument as above), so there are at least $i$ A-signal vertices. Analogously, there are at least $j$ B-signal vertices. Now if we can show that there are at most $i+j$ signal vertices, then we are done. 
	
	Observe that there is at least one buffer vertex within a distance $\eps$ to $g_0$ in between signal vertices that are matched to letter gadget in $A_i$ or $B_j$, as such a vertex must cover a $g_0^\beta$ subcurve between the letter gadgets. We call signal vertices that are matched to the same letter gadget in either $A_i$ or $B_j$ a group. (Note that by definition, a signal vertex cannot be matched to letter gadgets in both $A_i$ and $B_j$) This means that there are at least $i$ groups of A-signal vertices and at least $j$ groups of B-signal vertices.
	
	When matching $c^*$ and $\gamma(s)$ for some $s\in S$, we can only match at most $|s|$ groups of signal vertices to a $g_a$ or $g_b$ vertex in a letter gadget in $\gamma(s)$. So, for the at least $i+j-|s|$ remaining groups of signal vertices, we can either match them to a $g_0$ vertex in $\gamma(s)$, or to a corresponding $g_a$ or $g_b$ vertex. In the latter case, the signal vertex is matched to the same $g_a^\beta$ or $g_b^\beta$ subcurve in $\gamma(s)$ as another signal vertex in a different group. This means that the buffer vertex that separates the two signal vertices is matched to a $g_a$ or $g_b$ vertex in the letter gadget. So in all cases, we match two vertices at distance at leasts $1-\eps$. Since we do this for at least $i+j-|s|$ vertices, $\dtw_p(c^*,\gamma(s))\geq (1-\eps)(i+j-|s|)^{1/p}$.
	
	Now, we have 
	\begin{align*}
	\alpha (\dtw_p^q(c^*,A_i)+\dtw_p^q(c^*,B_j)) &\leq r -\sum_{s\in S} \dtw_p^q(c^*,\gamma(s))\\ 
	&\leq r - \sum_{s\in S} (1-\eps)^q(i+j-|s|)^{q/p}\\
	&= \alpha (i^{q/p}+j^{q/p}) + \sum_{s\in S} (1-(1-\eps)^q)(i+j-|s|)^{q/p}\\ 
	&\leq \alpha (i^{q/p}+j^{q/p}) + (1-(1-\eps)^q)|S|(i+j)^{q/p},
	\end{align*}
	so that $\dtw_p^q(c^*,A_i)+\dtw_p^q(c^*,B_j) \leq i^{q/p}+j^{q/p} + (1-(1-\eps)^q)(i+j)^{q/p} < i^{q/p}+j^{q/p} + \frac{1}{2}\min_{x\in \{i,j\}} (x+1)^{q/p} - x^{q/p}$.
	This means that there are at most $i+j$ signal vertices: suppose there are at least $i+1$ A-signal vertices, then $\dtw_p^q(c^*,A_i)+\dtw_p^q(c^*,B_j)\geq (1-\eps)^q ((i+1)^{q/p}+j^{q/p}) \geq i^{q/p}+j^{q/p} + ((i+1)^{q/p} - i^{q/p})/2$, a contradiction. Analogously, at least $j+1$ B-signal vertices lead to a contradiction. 
\end{proof}

\begin{theorem}\label{thm:DTW-hard}
	The average curve problem for the $(p,q)$-DTW distance is NP-hard, for any $p,q\in \mathbb{N}$. When parametrized in the number of input curves $m$, this problem is W[1]-hard. There exists no $f(n)\cdot n^{o(m)}$ time algorithm for this problem unless ETH fails.
\end{theorem}
\begin{proof}
	By Lemmas~\ref{lem:FCCS-to-DTW} and~\ref{lem:pq-dtw-to-FCCS}, we have a valid reduction from FCCS to the average curve problem.
	Since this reduction runs in polynomial time and FCCS is NP-hard (Lemma~\ref{lem:SCS-prime}), the average curve problem for discrete and continuous Fr\'echet is NP-hard. Since the reduction runs in polynomial time (note that $1/\eps$ can be bounded by a polynomial function in $n$, since $p,q$ are constants, so $\beta$ can be polynomially bounded) and the number of input curves is bounded by a linear function in $|S|$, the claim follows.
\end{proof}
\section[Algorithms for (k,l)-center and -median curve clustering]{Algorithms for $(k,\ell)$-center and -median curve clustering}

\subsection[(1+eps)-approximation for (k,l)-center clustering for discrete Frech\'et in R\^{}d]{$(1+\eps)$-approximation for $(k,\ell)$-center clustering for discrete Frech\'et in $\rr^d$}\label{sec:approx-center}

In this section, we develop a $(1+\eps)$-approximation algorithm for the $(k,\ell)$-center problem under the discrete Fr\'echet distance that runs in $O(mn\log (n))$ time for fixed $k,\ell,\epsilon$. In this algorithm, we use hypercube grids $L_v(a,b)$ around a vertex $v$ of width $a$ and resolution $b$: take the axis-parallel $d$-dimensional hypercube centered at $v$ of side-length $a$. Divide this hypercube into smaller hypercubes of side-length at most $b$. The grid $L_v(a,b)$ is the set of all vertices of the smaller hypercubes that intersect the ball of diameter $a$ around $v$. See figure~\ref{fig:hypergrid} for an example.

The algorithm is as follows:
First, we compute a set of curves $\mathcal{C}=\{c_1,\ldots, c_k\}$ that forms a $3$-approximation for the $(k,\ell)$-center problem, using the algorithm by Buchin et~al.~\cite{KAJFrechetCenter}. Let $\Delta$ be the cost of $\mathcal{C}$. 

\begin{figure}
	\centering
	\includegraphics{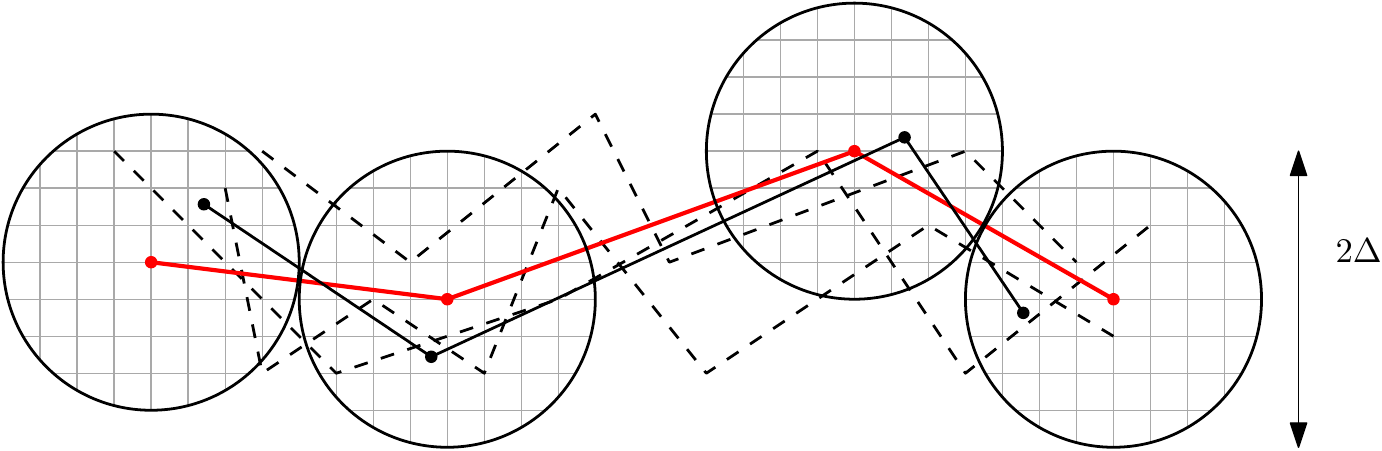}
	\caption{\label{fig:hypergrid} Given an approximate $(1,\ell)$-center curve (red) for a set of curves (dashed), the vertices of the optimal center curve (black) will be close to the hypercube grids around the vertices of the approximate center.}
\end{figure}

Let $V$ be the union of the hypercube grids $L_v(4\Delta, \frac{2\Delta\eps}{3\sqrt{d}})$ over all vertices $v$ of curves in $\mathcal{C}$. For every set of $k$ center curves with complexity $\ell$ using only vertices from $V$, compute the clustering and cost as centers for $G$, and return the set with minimal cost.

In order to show this algorithm gives an $(1+\eps)$-approximation, we use the following lemma to show that there is a set of $k$ center curves that is close enough to the optimal solution:

\begin{lemma}\label{lem:grid-search}
	Let $k,\ell\in \mathbb{N}$, $\delta\in \rr$ and $X>0$. Suppose there are two sets $\mathcal{C}=\{c_1,\ldots,c_k\}$ and $\mathcal{C}^*=\{c^*_1,\ldots,c^*_k\}$, both containing $k$ curves in $\rr^d$ of complexity $\ell$.
	Additionally, suppose that for all curves $c^*\in \mathcal{C}^*$, there exists a curve $c\in \mathcal{C}$ such that $\dfrech(c,c^*)\leq \delta$. Let $V=\{ L_v(2\delta,2\frac{X}{\sqrt{d}})\mid \text{$v$ is a vertex of a curve in $\mathcal{C}$} \}$.
	Then there is a set of curves $\widetilde{\mathcal{C}}=\{\tilde{c}_1,\ldots,\tilde{c}_k\}$, using only vertices from $V$, such that  
	$\dfrech(c^*_i,\tilde{c}_i) \leq X$, $|\tilde{c}_i|= \ell$, for all $1\leq i\leq k$.
\end{lemma}
\begin{proof}
	Let $v$ be a vertex of a curve in $\mathcal{C}$, and let $p$ be a point such that $\|p-v\|\leq \delta$. Then $p$ lies inside one of the small hypercubes and so there is a vertex $p'\in L_v(2\delta,2\frac{X}{\sqrt{d}})$ (a vertex of that small hypercube) such that $\|p-p'\|\leq \frac{\sqrt{d}}{2}\cdot \frac{2X}{\sqrt{d}}= X$.
	
	Let $c^*\in \mathcal{C}^*$. There exists a curve $c\in \mathcal{C}$ with $\dfrech(c,c^*)\leq \delta$, which means that each vertex $u$ of $c^*$ has distance at most $\delta$ to some vertex $v$ of $c$. So, there exists a vertex $v'\in L_v(2\delta,2\frac{X}{\sqrt{d}})$ such that $\|u-v'\|\leq X$. Construct the curve $\tilde{c}$ by connecting all such vertices $v'$ by line segments. By construction, $\dfrech(\tilde{c},c^*)\leq \delta$, $|\tilde{c}|=\ell$, and all vertices of $\tilde{c}$ are in $V$. So, we can take $\widetilde{\mathcal{C}}=\{\tilde{c}\mid c^*\in \mathcal{C}^*\}$.
\end{proof}

\begin{theorem}\label{thm:Frechet-center-approx}
	Given $m$ input curves in $\rr^d$, each of complexity at most $n$, and positive integers $k,\ell$ and some $0<\eps\leq 1$, we can compute an $(1+\eps)$-approximation to the $(k,\ell)$-center problem for the discrete Fréchet distance in $O\left(((C k\ell)^{k\ell}+\log(\ell+n))\cdot k\ell \cdot mn\right)$ time, with $C= \left(\frac{6\sqrt{d}}{\eps}+1\right)^d$.
\end{theorem}
\begin{proof}
	We first show that the algorithm above achieves this approximation ratio. Let $\mathcal{C}^*$ be an optimal optimal solution for the $(k,\ell)$-center problem, and $O$ its cost. Let $c^*\in \mathcal{C}^*$, then there is a curve $g\in G$ such that that $\dfrech(c^*,g)\leq O$ (assuming without loss of generality that its cluster is non-empty). Since the solution $\mathcal{C}$ has cost $\Delta$, there is a $c\in \mathcal{C}$ such that $\dfrech(c,g)\leq \Delta$. So, $\dfrech(c,c^*)\leq \dfrech(c,g)+\dfrech(g,c^*)\leq 2\Delta$, and by Lemma~\ref{lem:grid-search} with $\delta=2\Delta$ and $X=\eps\cdot \Delta/3 \leq \eps O$, there is a solution $\widetilde{\mathcal{C}}$ with the properties in the Lemma. Since for any $g\in G$, there is a curve $c^*\in \mathcal{C}^*$ such that $\dfrech(g,c^*)\leq O$, there is a $\tilde{c}\in \widetilde{\mathcal{C}}$ such that $\dfrech(g,\tilde{c})\leq \dfrech(g,c^*)+\dfrech(\tilde{c},c^*)\leq (1+\eps)O$. Since the algorithm returns the best solution using only from $V$, it returns a solution of cost at most that of $\widetilde{\mathcal{C}}$, and is therefore an $1+\eps$-approximation.
	
	For the running time, computing the $3$-approximation $\mathcal{C}$ takes $O(k\ell mn\log (\ell+n))$ time~\cite{KAJFrechetCenter}. A grid $L_v(a,b)$ has at most $(\lceil\frac{a}{b}\rceil+1)^d$ vertices and the curves in $\mathcal{C}$ have at most $k\ell$ vertices, so $|V|\leq k\ell (\lceil\frac{6\sqrt{d}}{\eps}\rceil+1)^d$. There are $O(|V|^{k\ell})$ solutions using only vertices from $V$, and we can test each solution in $O(k\ell mn)$ time: computing the discrete Fr\'echet distance between an input curve and a center curve takes $O(\ell n)$ time using dynamic programming, which we do for all $km$ pairs of input and center curves. In total, we get a running time of $O\left((|V| ^{k\ell}+\log(\ell+n))\cdot k\ell mn\right)$.
\end{proof}
Note that we can use any $\alpha$-approximation algorithm instead of the $3$-approximation algorithm by Buchin~et~al.~\cite{KAJFrechetCenter}, if we scale the grids accordingly. This changes the value of $C$ to $\left(\frac{2\alpha\sqrt{d}}{\eps}+1\right)^d$. If $\eps$ is very small, we can use this to get a smaller $C$ constant by running our algorithm twice, first computing a $1.01$-approximation, and using the approximation to compute the $(1+\eps)$-approximation.

When $\eps$ and $d$ are fixed constants, the algorithm from Theorem~\ref{thm:Frechet-center-approx} is fixed parameter tractable in the parameter $k+\ell$. In this setting, there is no $(1+\eps)$-approximation algorithm that is fixed parameter tractable in either $k$ or $\ell$ separately (the problem is not even in XP, in fact), unless $\textsc{P}=\textsc{NP}$. If we do not fix $\ell$, then achieving an approximation factor strictly better than $2$ is already NP-hard when $k=1$ and $d=1$~\cite{KAJFrechetCenter}. If we do not fix $k$ and if $\ell=1$, the $(k,\ell)$-center problem for discrete Fr\'echet is equivalent to the Euclidean $k$-center problem, which is NP-hard to approximate within a factor of $1.82$ for $d\geq 2$~\cite{feder1988optimal}.

\subsection[Approximation algorithms for (k,l)-median clustering for the discrete Fr\'echet distance in R\^{}d]{Approximation algorithms for $(k,\ell)$-median clustering for the discrete Fr\'echet distance in $\rr^d$}\label{sec:approx-k-l-median}

We construct an $(1+\eps)$-approximation for the $(k,\ell)$-median problem for the discrete Fr\'echet distance with a similar approach as above: first compute an constant factor approximation, and then search in hypercube grids around the vertices of that approximation. The algorithm for the constant factor approximation is essentially the same as the approximation algorithm from~\cite{driemel2016clustering} for 1D curves, except we use different subroutines and derive a tighter approximation bound.
We first introduce some techniques we will use to get a $12$-approximation. Given a polygonal curve $\gamma$, a \emph{simplification} is a polygonal curve that is similar to $\gamma$, but has only a few vertices. Specifically, \emph{a minimum error $\ell$-simplification} $\bar{\gamma}$ of a curve $\gamma$ is a curve of complexity at most $\ell$ that has a minimum distance to $\gamma$ among all curves with complexity at most $\ell$. We can compute a minimum error $\ell$-simplification under the discrete Fr\'echet distance for a curve $\gamma$ of complexity $n$ in $O(n\ell \log n \log (n/\ell)$ time~\cite{bjwyz-spcdfd-08}.

The $12$-approximation algorithms goes as follows: First, compute a minimum error $\ell$-simplification $\bar{g}$ for each input curve $g$ and let $\overline{G}$ be the set of all simplified curves. Then, compute a $4$-approximation for the $k$-median problem with $F=\overline{G}$ and $C=G$, using the algorithm by Jain~et~al.~\cite{JainMS02}. This yields a $12$-approximation:

\begin{theorem}\label{thm:12-approx-k,l-median}
	Given $m$ input curves in $\rr^d$, each of complexity at most $n$, and positive integers $k,\ell$, we can compute a $12$-approximation to the $(k,\ell)$-median problem for the discrete Fréchet distance in $O(m^3 + mn\ell(m+\log n \log (n/\ell)))$ time.
\end{theorem}
\begin{proof}
	We first show the approximation ratio. Let $\mathcal{C}^*$ be the optimal solution to the $(k,\ell)$-median problem with cost $O$, and let $\mathcal{C}$ be the solution computed by our algorithm above. Each center curve $c^*_i$ has a set $G^*_i\subseteq G$ as its cluster. Let $c_i'$ be the minimum error $\ell$-simplification of a curve $c_i$ from $G^*_i$ that has minimum distance to $c_i^*$. The curves $\mathcal{C}'=\{c_1',\ldots, c_k'\}$ are a $3$-approximation to the $(k,\ell)$-median problem: we have $\sum_{g\in G} \min_{i=1}^k \dfrech(g,c_i') \leq \sum_{i=1}^k \sum_{g\in G^*_i} \dfrech(g,c_i') \leq \sum_{i=1}^k \sum_{g\in G^*_i} \dfrech(g,c^*_i) + \dfrech(c^*_i,c_i) + \dfrech(c_i,c_i') \leq 3 \sum_{i=1}^k\sum_{g\in G^*_i} \dfrech(g,c^*_i) = 3O$, where $\dfrech(c_i',c_i)\leq \dfrech(c^*_i,c_i)$ because $|c^*_i|=\ell$ and $c_i'$ is a minimum error $\ell$-simplification of $c_i$, and $\dfrech(c_i,c^*_i)\leq \dfrech(g,c^*_i)$ for all $g\in G^*_i$ by definition of $c_i$.
	$\mathcal{C}'$ is some solution to the $k$-median problem with $F=\overline{G}$ and $C=G$ of cost at most $3O$, so the optimal solution to this problem has cost at most $3O$. Since we compute a $4$-approximation for that problem, the result has cost at most $12O$.
	
	For the running time, note that computing the simplification of all curves in $G$ takes $O(mn\ell \log n \log (n/\ell))$ time. Then, we can compute the discrete Fr\'echet distances between pairs from $\overline{G}\times G$ in $O(m^2\cdot \ell n)$ time, and run the algorithm by Jain~et~al.~\cite{JainMS02} in $O(m^3)$ time.
\end{proof}

We can modify the algorithm above to run in $\widetilde{O}(mn)$ time when $k,\ell$ are constant: Compute $\overline{G}$ as before, but now use the algorithm by Chen~\cite{Chencoresets09} to compute a $10.5$-approximation to the $k$-median problem with $F=C=\overline{G}$. This gives a $42$-approximation:
\begin{lemma}\label{lem:12-approx-k,l-median}
	Given $m$ input curves in $\rr^d$, each of complexity at most $n$, and positive integers $k,\ell$, we can compute a $42$-approximation to the $(k,\ell)$-median problem for the discrete Fréchet distance in $O(mn\ell \log n \log (n/\ell) + \ell^2(mk+k^7\log^5 m))$ time.
\end{lemma}
\begin{proof}
    The proof is similar to Theorem~\ref{thm:12-approx-k,l-median}, but now simplifications are clustered instead of the original curves. We first show the approximation ratio. Given a cluster $G^*_i\subset \overline{G}$ from the optimal clustering with center $c_i^*$, let $\bar{c}_i$ be the simplification of a curve $g$ in this cluster such that $\dfrech(\bar{c}_i,c_i^*)$ is minimal. The curves $\overline{\mathcal{C}}=\{\bar{c}_1,\ldots, \bar{c}_k'\}$ are a $4$-approximation to the $(k,\ell)$-median problem: we have $\sum_{g\in G} \min_{i=1}^k \dfrech(\bar{g},\bar{c}_i) \leq \sum_{i=1}^k \sum_{g\in G^*_i} \dfrech(\bar{g},\bar{c}_i) \leq \sum_{i=1}^k \sum_{g\in G^*_i} \dfrech(\bar{g},c^*_i) + \dfrech(c^*_i,\bar{c}_i) \leq \sum_{i=1}^k\sum_{g\in G^*_i} 2\dfrech(\bar{g},c^*_i) \leq 2\sum_{i=1}^k\sum_{g\in G^*_i} \dfrech(\bar{g},g)+\dfrech(g,c^*_i) \leq 2 \sum_{i=1}^k\sum_{g\in G^*_i} 2\dfrech(g,c^*_i) = 4O$, where $\dfrech(\bar{c}_i,c^*_i)\leq \dfrech(\bar{g},c^*_i)$ by definition of $\bar{c}_i$ and $\dfrech(\bar{g},g) \leq \dfrech(g,c^*_i)$ because $|c^*_i|=\ell$ and $\bar{g}$ is a minimum error $\ell$-simplification of $g$. Since we compute a $10.5$-approximation to the problem for which $\overline{\mathcal{C}}$ is a solution, the approximation ratio $10.5\cdot 4 = 42$.
    
    Computing the simplification of all curves in $G$ takes $O(mn\ell \log n \log (n/\ell))$ time. The algorithm by Chen~\cite{Chencoresets09} takes $O(mk+k^7\log^5 m)$ time, so it uses at most that number of distance computations between curves in $\overline{G}$, which take $O(\ell^2)$ time each.
\end{proof}

We now use the $42$-approximation algorithm to compute an $(1+\eps)$-approximation. Let $\mathcal{C}=\{c_1,\ldots, c_k\}$ be the solution given by the approximation algorithm above, and $\Delta$ its cost. If $k=1$, let $V$ be the union of the hypercube grids $L_v(4\Delta/m, \frac{\eps\Delta}{21m\sqrt{d}})$ over all vertices $v$ of curves in $\mathcal{C}$. If $k>1$, let $V$ be the union of the grids $L_v(4\Delta, \frac{\eps\Delta}{21m\sqrt{d}})$ over the same vertices, instead. For every set of $k$ center curves with complexity $\ell$ using only vertices from $V$, compute the clustering and cost (using the median objective) as centers for $G$, and return the set with minimal cost.

\begin{theorem}\label{thm:Frechet-median-approx}
	Given $m$ input curves in $\rr^d$, each of complexity at most $n$, and positive integers $k,\ell$ and some $0<\eps\leq 1$, we can compute an $(1+\eps)$-approximation to the $(k,\ell)$-center problem for the discrete Fréchet distance in $O\left( mn\ell ((C \ell)^{\ell} + \log n \log (n/\ell)) \right)$ time when $k=1$ with $C=\left(\frac{84\sqrt{d}}{\eps}\right)^d$. When $k>1$, we require\\  $O\left((C k\ell)^{k\ell}\cdot k\ell \cdot m^{dk\ell+1}n + mn\ell \log n \log (n/\ell) + \ell^2(mk+k^7\log^5 m)\right)$ time.
\end{theorem}
\begin{proof}
	We first show the approximation ratio. Let $\mathcal{C}^*=\{c_1^*,\ldots, c_k^*\}$ be an optimal solution for the $(k,\ell)$-median problem, $G^*_i\subset G$ the cluster induced by the center $c_i^*$, and $O$ the total cost of this solution. Let $\widetilde{\mathcal{C}}= \{\tilde{c}_1, \ldots, \tilde{c}_k\}$ be a set of curves with complexity at most $\ell$ such that for all $1\leq i \leq k$, there is a curve $\tilde{c}_{j}\in \widetilde{\mathcal{C}}$ with $\dfrech(c_i^*,\tilde{c_j})\leq \eps O/m$.
	Since $\sum_{g\in G} \min_{j=1}^k \dfrech(g,\tilde{c}_j) \leq \sum_{i=1}^k \sum_{g\in G_i^*} \dfrech(g, \tilde{c}_j) \leq \sum_{i=1}^k \sum_{g\in G_i^*} \dfrech(g, c^*_i) + \dfrech(c^*_i, \tilde{c}_i)\leq  \sum_{g\in G}\min_{i=1}^k \dfrech(g, c^*_i) + \eps O/m = (1+\eps)O$, the set $\widetilde{\mathcal{C}}$ is an $(1+\eps)$-approximation. We will show that there is such a set that uses only vertices of $V$. 
	
	If $k=1$, then  $\dfrech(c_1,c^*_1)= \frac{1}{m}\sum_{g\in G} \dfrech(c_1,c^*_1)\leq \frac{1}{m}\sum_{g\in G} \dfrech(c_1, g) + \dfrech(g, c^*_1) \leq (\Delta+O)/m \leq 2\Delta/m$. Applying Lemma~\ref{lem:grid-search} with $\delta=2\Delta/m$ and $X = \eps\Delta/(12m)\leq \eps O/m$, there is a $(1+\eps)$-approximation using only vertices of $V$.
	
	Otherwise, if $k>1$, then for each $c^*_i$ there is a $c_j$ such that the clusters of these centers share some curve $g\in G$. So, $\dfrech(c^*_i,c_j) \leq \dfrech(c^*_i, g) + \dfrech(g, c_j) \leq O+\Delta\leq 2\Delta$. Applying Lemma~\ref{lem:grid-search} with $\delta=2\Delta$ and $X = \eps\Delta/(12m)\leq \eps O/m$, there is a $(1+\eps)$-approximation using only vertices of $V$.
	
	For the running time, we have $|V| \leq k\ell (\lceil\frac{a}{b}\rceil+1)^d$ when we use grids with width $a$ and resolution $b$. If $k=1$, $\frac{a}{b} = \frac{4\Delta/m}{\eps\Delta/(21m\sqrt{d})} = \frac{84\sqrt{d}}{\eps}$. If $k>1$, $\frac{a}{b} = \frac{84m\sqrt{d}}{\eps}$. The rest of the analysis is similar to that in Theorem~\ref{thm:Frechet-center-approx}.
\end{proof}

The additional $m^{dk\ell}$ factor in our algorithm when $k>1$ is due to the following:
We need to find curves with distance at most $\eps\Delta/(42m)$ to an optimal solution in order to get an $(1+\eps)$-approximation. When $k=1$, the number of grid points is still independent of $m$, since the distance of the curve in $\mathcal{C}$ to an optimal curve is $O(\Delta/m)$. However, when $k>1$, it is possible that all optimal solutions have a cluster of constant size. Then the center curve of that cluster can have distance $\Omega(\Delta)$ to all curves in $\mathcal{C}$. 

\subsection[Exact algorithm for (k,l)-center under discrete Fr\'echet in 2D]{Exact algorithm for $(k,\ell)$-center under discrete Fr\'echet in 2D}

We give an algorithm that solves the $(k,\ell)$-center problem for the discrete Fréchet distance in 2D in polynomial time for fixed $k$ and $\ell$. We first show how to solve the decision version of this problem and use it as a subroutine to solve the optimisation problem.

The main idea of the algorithm for the decision version is based on the following observation: for a given $r$, we have  $\min_{c\in\mathcal{C}}\dfrech(c,g)\leq r$ for all $g\in G$ if and only if each vertex $p$ of a curve in $\mathcal{C}$ lies in the intersection of the disks of radius $r$ around all vertices $q$ from curves in $G$ that $p$ is matched with. Furthermore, it does not matter where the vertex $p$ lies within the intersection region. This means we can select a vertex for each maximal overlapping region (i.e. each region such that the set of disks intersecting the region is not contained in another region) and exhaustively test all sets with $k$ curves of $\ell$ vertices that can be constructed by using only the selected vertices to determine if there exists a set of curves $\mathcal{C}$ such that $\min_{c\in\mathcal{C}}\dfrech(c,g)\leq r$ for all $g\in G$.

\begin{figure}[t]
	\centering
	\includegraphics[scale=0.7]{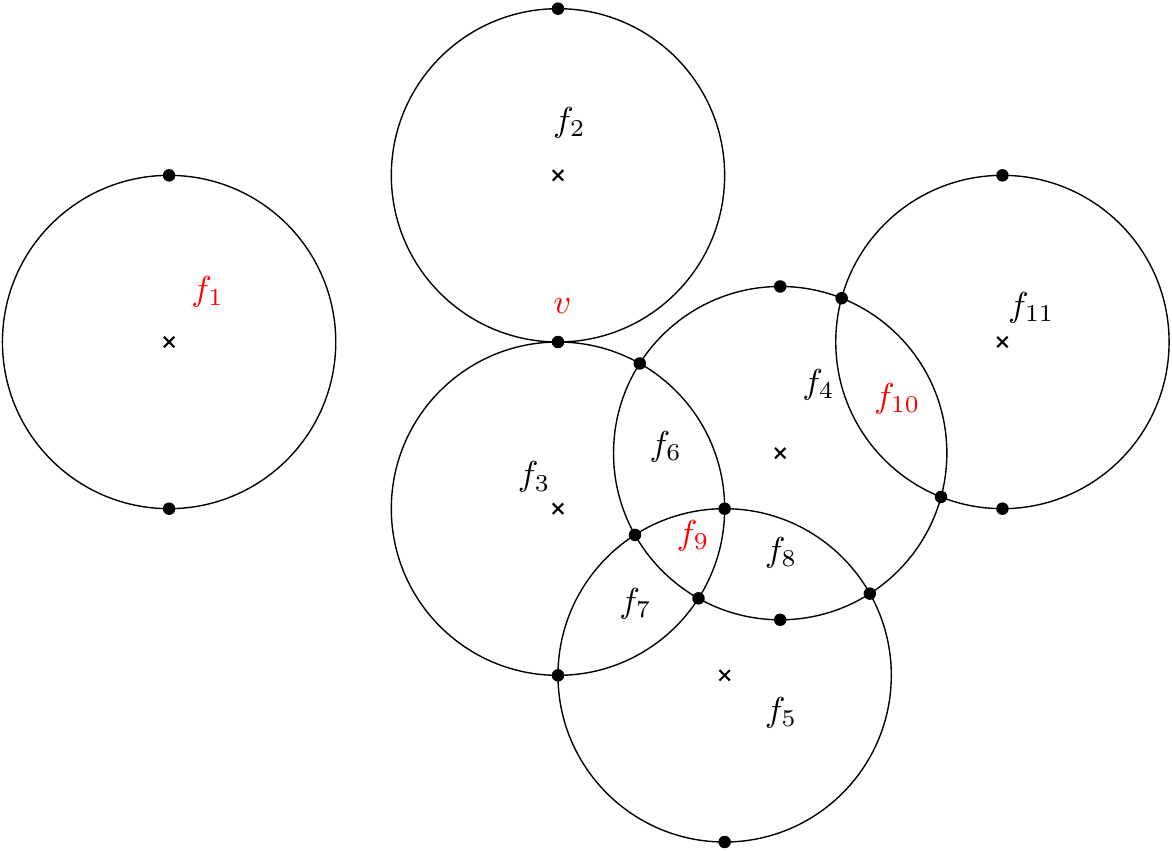}
	\caption{\label{fig:arc-graph-example} An example configuration of $\mathcal{G}=(V,E)$. Crosses indicate the vertices from the curves in $G$, dots indicate vertices from $V$ and all bounded faces are numbered. The maximal intersection regions are the faces $f_1$ and $f_9$ and the vertex $v$ (in red). Note that while all arcs on the boundary of $f_2$ are convex for that face, $f_2$ is not maximal, since its boundary intersects the boundary of $f_3$ only at vertex $v$.}
\end{figure}

To find all maximal intersection regions, we first compute the planar graph $\mathcal{G}=(V,E)$, where $V$ is the set of all intersection points between boundaries of disks centred around a vertex from our input curves with radius $r$ and $E$ is the set of arcs on the boundary of those disks ending at two intersection points. This graph has $O((nm)^2)$ vertices and arcs and can be computed in $O((nm)^2)$ time~\cite{chazelle1986circle}, see Figure~\ref{fig:arc-graph-example} for an example.

By traversing the intersection points and arcs on the boundary, we can find the at most $O((nm)^2)$ maximal intersection regions. So, we test $O((mn)^{2k\ell})$ sets of center curves, for which we can test whether an input curve has discrete Fr\'echet distance less than $r$ to a single curve among the $k$ center curves in $O(m\ell)$. This means the algorithm for the decision version takes $O((mn)^{2k\ell}k\ell m)$ time.

To find a minimum $r$ such that a $(k,\ell)$-center exists, note that we only have to consider the decision problem for those $r$ where the topology of the intersection regions in $\mathcal{G}$ is different. If we start with $r=0$ and gradually increase it, the topology of $\mathcal{G}$ changes only when a new maximal intersection is created, which then consists of exactly one point $p$. This means that there is a subset of our disks such that point $p$ is the earliest point where all disks have a non-empty intersection. So, $p$ must be the center of the minimum enclosing disk for this subset of disks. Since a minimum enclosing disk is determined by at most $3$ points, there can be at most one unique point for every triple in set of vertices of the input curves which give at most $O((mn)^3)$ distinct values of $r$ where the topology of $\mathcal{G}$ changes. By performing a binary search on these values, we can find the optimal value in $O(\log(mn))$ calls to the algorithm for the decision version and we get the following result:

\begin{theorem}\label{thm:Frechet-exact}
	Given a set of $n$ curves $G$ in the plane with at most $m$ vertices each, we can find a solution to the $(k,\ell)$-center problem for the discrete Fréchet distance in $O((mn)^{2k\ell}k\ell m\log (mn))$ time. 
\end{theorem}
\section{Conclusion}
    In this paper, we have shown that the $1$-median problem is computationally hard under the discrete Fr\'echet, continuous Fr\'echet, and DTW distance. A natural question is whether this problem is hard to approximate. Efficient constant factor approximation algorithms are known for the Fr\'echet distance (see Section~\ref{sec:approx-k-l-median}), but not for DTW. If we extend our analysis in Lemma~\ref{lem:curve-frechet-to-FCCS} to a solution $c^*$ with cost $(1+\eps) r$ for some $\eps> 0$, we can show  $\dfrech(c^*,g) \leq 1 + O(\eps m)$ for all input curves $g$ (where the constant is independent of other input parameters). Together with the approximation lower bound of $2$ for $1$-center under continuous Fr\'echet distance~\cite{Struijs2018Curve}, this implies a lower bound of $1 + \Omega(\frac{1}{m})$ on the approximation factor for $1$-median. If we do the same for Lemma~\ref{lem:pq-dtw-to-FCCS}, we get that it is hard to approximate $1$-median under $(p,q)$-DTW for any factor $< 1 + 2((1+\frac{1}{\min(i,j)})^{q/p}-1)$. So, it remains an open problem to find a constant lower bound for approximating $1$-median for this distance measures. 

On the positive side, we have given $(1+\eps)$-approximation algorithms for $(k,\ell)$-center and $(k,\ell)$-median problems under discrete Fr\'echet in Euclidean space and an exact algorithm for the $(k,\ell)$-center problem under discrete Fr\'echet in 2D that all run in polynomial time for fixed $k,\ell,\eps$. 
It would be interesting to see if these algorithms can be adapted to the DTW or continuous Fr\'echet settings. Our approximation algorithms rely on the fact that good approximations have small distance to some optimal solution and that we can search a bounded space (the set of balls surrounding the vertices) for better approximations. The first property does not hold for DTW, since it is non-metric and the second property does not hold for continuous Fr\'echet, since the vertices of a curve with small continuous Fr\'echet distance do not have to be near the vertices of the other curve. The latter property is also crucial for the exact algorithm.

\bibliographystyle{plainurl}
\bibliography{refs}

\end{document}